\documentclass[letterpaper, 10 pt, conference]{ieeeconf}  

\IEEEoverridecommandlockouts                              

\usepackage{graphicx}
\usepackage{subfig}
\usepackage{amsmath} 
\usepackage{amssymb}  
\usepackage{dsfont}

\usepackage{algorithm}
\usepackage{algpseudocode}

\usepackage{amsthm}

\usepackage[noadjust]{cite}

\newtheorem{theorem}{Theorem}
\newtheorem{definition}{Definition}
\newtheorem{proposition}{Proposition}
\newtheorem{lemma}[theorem]{Lemma}
\newtheorem{assumption}{Assumption}
\let\origtheassumption\theassumption
\newtheorem*{remark}{Remark}

\let\labelindent\relax
\usepackage{enumitem}

\title{\LARGE \bf
Estimation of Unknown Payoff Parameters in Large Network Games  
}

\author{Feras Al Taha$^{1}$ and Francesca Parise$^{1}$
\thanks{This work was supported in part by the Natural Sciences and Engineering Research Council of Canada.}
\thanks{$^{1}$Feras Al Taha and Francesca Parise are with the School of Electrical and Computer Engineering,
        Cornell University, NY 14850, USA
        {\tt\small \{foa6,fp264\}@cornell.edu}}%
}

\begin{document}

\maketitle
\thispagestyle{empty}
\pagestyle{empty}

\begin{abstract}

We consider network games where a large number of agents interact according to a network sampled from a random network model, represented by a graphon. 
By exploiting previous results on convergence of such large network games to graphon games, we examine a procedure for estimating unknown payoff parameters, from observations of equilibrium actions, without the need for exact network information. 
We prove smoothness and local convexity of the optimization problem involved in computing the proposed estimator.
Additionally, under a notion of graphon parameter identifiability, we show that the optimal estimator is globally unique.
We present several examples of identifiable homogeneous and heterogeneous parameters in different classes of linear quadratic network games with numerical simulations to validate the proposed estimator. 

\end{abstract}

\section{INTRODUCTION}

Systems involving very large numbers of autonomous agents making strategic decisions and influencing each other over a network structure are becoming ubiquitous.
For example, they appear in applications involving power and traffic networks in engineering settings as well as product adoption, targeted marketing, and opinion dynamics in socio-economic settings. 
Studying how agents make decisions in these complex environments is a fundamental prerequisite for the successful design of interventions and control laws aimed at improving  welfare or system performance. 
To this end, game theoretical principles are typically used to model agents' decisions via payoff maximization, resulting in the concept of Nash equilibrium (i.e., a set of actions in which no agent has interest in unilateral deviations) as a solution outcome.
When translating these results to practice, however, a main issue emerges: while the parametric form of agents' payoff functions might be known, in most applications  the parameters themselves are not.
For example, in games capturing agents' decisions under peer pressure, the strength of neighbors' peer effect on an individual's marginal return might not be known \cite{JACKSON201595} and, in fact, may vary for different network instances (e.g. different schools, neighborhoods, etc.).
In these settings, it is then of paramount importance to understand whether a central planner can estimate the unknown parameters from observations of agents' actions at equilibrium.
This ability would indeed enable the central planner to design interventions steering agents towards equilibria with improved welfare or system efficiency \cite{parise2019graphon,parise2021analysis,galeotti2020targeting}.

Starting from the seminal work of Bramoullé et al. \cite{bramoulle2009identification}, a large literature studied the above question under the assumption that the planner knows  the network over which  agents interact.
When we turn our attention to applications involving a large number of agents, however, collecting data about the exact network of interactions can become very expensive or not at all possible because of privacy concerns. 
Consequently, recent works started investigating parameter estimation  under  partial or statistical network information \cite{de2018recovering,boucher2020estimating,lewbel2019social,chandrasekhar2011econometrics}.
Importantly, all the works cited above focus on the specific problem of estimating the peer effect parameter in linear in means models \cite{USHCHEV2020104969}.
The key objective of this paper is to develop a general parameter estimation procedure that: i) relies  only on statistical instead of exact information about network interactions and ii) can be applied for parameter estimation in generic network games.

To obtain such a result, we build on  the framework of graphon games  recently proposed in \cite{parise2019graphon}. 
Graphon games are  games played over a continuum of agents  that interact heterogeneously according to a graphon. 
Building on an interpretation of graphons as random network models  (which generalizes for example Erd\"os-R\'enyi  and  stochastic block models (SBM)~\cite{lovasz2012large}), \cite{parise2019graphon} shows that equilibria of  network games in which the network of interactions is sampled from the graphon (termed \textit{sampled network games}) converge, in the limit of large populations, to the equilibrium of the corresponding graphon game. Equilibria of graphon games can thus be seen as an approximation of strategic behavior in large network games, computed by using only information about the random network model.
Based on this result, \cite{parise2019graphon} suggests a novel procedure for payoff parameter estimation in sampled network games without the need for exact network data.
Specifically, given an observation of the equilibrium of a network game with unknown parameters, the proposed approach consists of selecting as estimator the parameters for which the equilibrium of the corresponding graphon game is closest to the observed equilibrium.
It is shown in \cite{parise2019graphon} that this estimator is asymptotically consistent if the parameter satisfies an identifiability assumption capturing games in which equilibria that are close are generated by parameters that are also close.

In this work, we address two main open problems related to the  estimation procedure detailed above. 
First, finding the parameter for which the graphon game equilibrium is the closest to the observed equilibrium requires the solution of an optimization problem. 
We here show that the objective function of such an optimization problem is smooth and locally strictly convex around the true parameter. 
Moreover, under the identifiability assumption above, we show that the optimization problem admits a unique global optimizer, thus guaranteeing that the proposed estimator is unique.
Second, we prove that the required identifiability assumption holds for several examples of linear quadratic (LQ) network games with both homogeneous and heterogeneous parameters.
We validate the convergence of the proposed estimator on these games with numerical simulations.

This paper is part of a growing literature that studies strategic behavior in large network games (see e.g.  \cite{grammatico2017proximal,zhu2022asynchronous, koshal2016distributed,belgioioso2020distributed,shokri2021network,parise2021analysis,parise2020distributed}), in particular using graphons \cite{caines2021graphon, carmona2022stochastic, gao2019graphon,avella2018centrality}.
However, none of the works cited above focuses on parameter estimation.

The rest of the paper is organized as follows.
Section~\ref{formulation} presents the network game setup and its connection to graphon games.
Section~\ref{problem} introduces the parameter estimation problem.  
Section~\ref{guarantees} provides our main result on properties of the corresponding optimization problem.
Section~\ref{LQ_games_unconstr} provides examples of identifiable parameters in different linear quadratic network games and Section~\ref{sec:sims} demonstrates the convergence of the estimator with numerical simulations.
Omitted proofs are given in the Appendix.

\textit{Notation:} We denote by $L^2([0, 1])$ the space of square integrable functions defined on $[0,1]$ and by $L^2([0,1];\mathbb{R}^n)$ the space of square integrable vector valued functions defined on $[0,1]$.
The norms on these spaces are  $\|v\|:=\sqrt{\sum_{i=1}^n v_i^2}$, $\|f\|_{L^2} := \sqrt{\int_0^1 f(x)^2 dx}$ and $\|g\|_{L^2;\mathbb{R}^n} := \sqrt{\int_0^1 \|g(x)\|^2 dx}$ where $v\in\mathbb{R}^n$, $f\in L^2([0,1])$ and $g \in L^2([0,1];\mathbb{R}^n)$. 
Additionally, we denote by $\|\cdot\|_{\infty}$ the uniform norm (or sup norm) of an operator.
We denote by $[v]_j$ the $j$th component of a vector $v\in\mathbb{R}^n$ and by $[A]_{ij}$ the $ij$th entry of a matrix $A\in\mathbb{R}^{m\times n}$.
The symbol $\mathds{1}$ denotes the vector of all ones (with appropriate dimension) and $\mathds{1}(x)$ the function constantly equal to one on the unit interval.
The symbol $\mathbb{I}$ denotes the identity operator.

\section{Recap on Finite and Infinite Network Games} \label{formulation}

\subsection{Finite network games} \label{sec:network_game}

Network games can be used to model settings where a \textit{finite number  of agents} interact strategically over a network. In the following, we represent the network of interactions with its adjacency matrix $P\in \mathbb{R}^{N\times N}$ with diagonal entries equal to zero (i.e., without self-loops) and assume that each agent $i\in\{1,...,N\}$ aims at selecting a scalar strategy\footnote{We assume scalar strategies for simplicity of exposition. Similar arguments can be made for  vector strategies.} $s^i \in \mathbb{R}$ in a feasible set $\mathcal{S} \subseteq \mathbb{R}$ to maximize a payoff function 
\begin{align}
    U(s^i, z^i(s), \theta^i) \label{network_payoff}
\end{align}
where $s:= [s^i]_{i=1}^N \in \mathbb{R}^N$ is the strategy profile, $z^i(s) := \frac{1}{N} \sum_{j=1}^N P_{ij} s^j$ denotes the local aggregate computed according to the network and $\theta^i \in \Theta \subseteq \mathbb{R}^m$ 
models heterogeneity in the payoff functions of different agents.
We remark that the local aggregate $z^i(s)$ of an agent does not include its own strategy $s^i$ (since $P_{ii}=0$).
The model is said to be homogeneous across agents when $\theta^i = \theta $ for all $i$.

\begin{definition}[Nash equilibrium]A strategy $\bar{\bar{s}} \in \mathbb{R}^N$ with associated local aggregate $\bar{\bar{z}}:=[\bar{\bar{z}}^i]_{i=1}^N$ where $\bar{\bar{z}}^i := z^i(\bar{\bar{s}})$ is a Nash equilibrium if for all $i\in\{1,\hdots,N\}$ , we have $\bar{\bar{s}}^i \in \mathcal{S}$ and
\begin{align*}
    U(\bar{\bar{s}}^i, \bar{\bar{z}}^i, \theta^i) \ge U(\tilde{s}, \bar{\bar{z}}^i,\theta^i) \qquad \textrm{for all } \tilde{s} \in \mathcal{S}.
\end{align*}
\end{definition}

\subsection{Graphon games}\label{sec:graphon}

A graphon game is defined in terms of a \textit{continuum of agents}, indexed by $x \in [0,1]$, that interact heterogeneously according to a symmetric and measurable graphon $W: [0,1]^2 \mapsto [0,1]$.  Intuitively, $W(x,y)$ measures the level of interaction between infinitesimal agents $x$ and $y$.
As in network games, the goal of each agent in a graphon game is to select a strategy $s(x)\in \mathcal{S}$ to maximize their payoff 
\begin{align}
    U(s(x), z(x|s), \theta(x)) \label{graphon_payoff}
\end{align}
where $z(x|s) := \int_0^1 W(x,y) s(y) dy$ is the local aggregate experienced by agent $x$ and $\theta : [0,1] \to \Theta \subseteq \mathbb{R}^m$ is a function modelling payoff heterogeneity across agents.
Note that the payoff function $U$ is the same payoff function as in \eqref{network_payoff}; 
the only difference is how the network aggregate is computed.

\begin{definition}[Graphon Nash equilibrium] A function $\bar{s}\in L^2([0,1])$ with associated local aggregate $\bar{z}(x):= z(x|\bar{s}) = \int_0^1 W(x,y) \bar{s}(y)dy$ is a Nash equilibrium for the graphon game if for all $x\in [0,1]$,  we have $\bar{s}(x)\in\mathcal{S}$ and
\begin{align*}
    U(\bar{s}(x),\bar{z}(x),\theta(x)) \ge U(\tilde{s},\bar{z}(x),\theta(x)) \quad \textrm{for all } \tilde{s} \in \mathcal{S}.
\end{align*}
\end{definition}

Conditions for existence and uniqueness of the graphon Nash equilibrium are derived in \cite{parise2019graphon} in terms of  properties of the graphon operator $\mathbb{W}: L^2([0,1]) \mapsto L^2([0,1])$ given by
\begin{align*}
    f(x) \mapsto (\mathbb{W} f)(x) = \int_0^1 W(x,y)f(y)dy,
\end{align*}
which intuitively plays the same role as the adjacency matrix for finite networks. These conditions are summarized in the following assumption.

\edef\oldassumption{\the\numexpr\value{assumption}+1}
\setcounter{assumption}{0}
\renewcommand{\theassumption}{\oldassumption.\alph{assumption}}

\begin{assumption}[Existence and uniqueness] \label{a:exist_uniq}
\begin{enumerate}[label=(\roman*)]
    \item []
    \item \label{a:smooth_convex} The function $U(s,z, \theta)$ in (\ref{graphon_payoff}) is continuously differentiable and strongly concave in $s$ with uniform constant $\alpha_U$ for all  $z$ and $\theta$.
Moreover, $\nabla_s U(s,z,\theta)$ is uniformly Lipschitz in $z$ and $\theta$ with constants $\ell_U, \ell_{\theta}$ for all $s$.
\item \label{a:strategy_set} The set $\mathcal{S}$ is convex and compact, so that $s_{max} := \max_{s\in \mathcal{S}} \|s\| < \infty$.
\item \label{a:contraction} 
The largest eigenvalue $\lambda_{\rm max}(\mathbb{W})$ of the graphon operator $\mathbb{W}$ satisfies the bound $\lambda_{\rm max}(\mathbb{W}) < \frac{\alpha_U}{\ell_U}$. 
\end{enumerate}
\end{assumption}
 
\begin{remark}
Under conditions \ref{a:smooth_convex} and \ref{a:strategy_set}, existence of a Nash equilibrium follows from standard fixed point argument.
Condition \ref{a:contraction} guarantees that the best response mapping is a contraction from which uniqueness follows. 
See Theorems~1 and 2 in \cite{parise2019graphon}.
\end{remark}

\subsection{Sampled network games} \label{sec:sampled_eq}

Besides being of interest as models of heterogeneous interactions in infinite populations, graphons can be used as random network models \cite{lovasz2012large}. 
Specifically, given any graphon~$W$, a \textit{sampled network} can be obtained by uniformly and independently sampling $N$ points\footnote{Without lost of generality, the points $\{t_i\}_{i=1}^N$ are assumed to be ordered such that $t_i\le t_{i+1}$, $i=1,\dots,N-1$, since the nodes can be relabeled.} $\{t_i\}_{i=1}^N$ from $[0, 1]$ and by defining a 0-1 adjacency matrix $P^{[N]}$ corresponding to a graph with $N$ nodes, no self-loops (i.e., $P_{ii}^{[N]}=0$ for all $i$) and random links $(i,j)$  sampled with Bernoulli probability $W(t_i,t_j)$. 
Graphons can therefore be used to encode statistical information about the likelihood of agents' interactions, with the understanding that the network observed in reality (i.e., $P^{[N]}$) is one possible realization of such random network model. 
Building on this statistical interpretation of graphons, \cite{parise2019graphon} shows that, for $N$ large enough, graphon Nash equilibria (as defined in Section \ref{sec:graphon}) are a good approximation of strategic behavior in any finite network game (as defined in Section~\ref{sec:network_game}) where agents interact over a network sampled from the graphon (which we term a \textit{sampled network game}\footnote{Since the maximum eigenvalue of the sampled network $P^{[N]}$ converges almost surely to the maximum eigenvalue of the graphon $W$, Assumption~\ref{a:exist_uniq} guarantees existence and uniqueness of equilibria in sampled network games for $N$ large enough with high probability.}).

\begin{remark}
Equilibria in finite network games are vectors instead of functions. 
To obtain comparable objects, we define a piecewise constant interpolation\footnote{Rather than interpolating the equilibria about the points $\{t_i\}_{i=1}^N$, we interpolate them about a regular grid $\{i/N\}_{i=1}^N$ so that the players' strategies are assigned equal weight.} of the network game equilibrium $\bar{\bar{s}}^{[N]}$ as a function $\bar s^{[N]}(x)=\left[\bar{\bar{s}}^{[N]}\right]_i$ for all $x\in\left[\frac{i-1}{N},\frac{i}{N}\right].$ 
In the following, we use the notation $\bar{\bar{s}}^{[N]}\in\mathbb{R}^N$ for a vector-valued equilibrium and $\bar s^{[N]}\in L^2([0,1]) $ for its interpolation.
\end{remark}

\begin{proposition}[{\cite[Theorem 2]{parise2019graphon}}] \label{prop:convergence}
Consider a graphon game satisfying Assumption \ref{a:exist_uniq} with unique Nash equilibrium $\bar{s}\in L^2([0,1])$. 
Let $\bar{s}^{[N]}\in L^2([0,1])$ be a piecewise constant interpolation of the equilibrium of a network game sampled from this graphon game. 
Then, $\| \bar{s}^{[N]} - \bar{s} \|_{L^2} \overset{a.s.}{\to} 0 \textrm{ as } N \to \infty$.
\end{proposition}

The key importance of this result is that the graphon equilibrium can be computed by relying only on information about the random network model, without the need for information about exact agent interactions ($P^{[N]}$). 
We next show how this key observation can be used for parameter estimation in settings in which the central planner does not have full network knowledge.

\section{The parameter estimation problem} \label{problem}

In many applications of interest, agents' payoffs may depend on parameters that are unknown to the central planner. 
In the following, we capture this aspect by assuming that the heterogeneity vector $\theta$, which characterizes agent-specific behavior, may depend on some unknown parameter $\eta \in \Xi \subseteq \mathbb{R}^n$ and we stress this dependence with the notation  $\theta_{\eta}$. This paper addresses the task of  identifying $\eta$ from the observation of a sampled equilibrium $\bar{\bar{s}}^{[N]} \in \mathbb{R}^N$ and the labels $\{t_i\}_{i=1}^N$, which are assumed to be known since they can represent an observable trait of the players (e.g., their community or geographical location). 
While most of the literature focused on settings in which $P^{[N]}$ is known, we here assume that the central planner cannot observe the sampled network, but instead has information about the random network model (i.e. the graphon).
With this information, the central planner can compute the graphon Nash  equilibrium corresponding to any possible choice of parameter $\eta$.
Building on Proposition~\ref{prop:convergence}, the central planner can then estimate the true parameter $\eta$ by choosing as estimator the parameter $\hat{\eta}$ which yields the closest graphon equilibrium to the observed equilibrium. 
Mathematically, we define the estimator
\begin{align}
    \hat{\eta} := \arg \min_{\eta \in \Xi} \quad \left \| \bar{s}^{[N]} - \bar{s}_\eta \right \|_{L^2}^2 \label{eta_estimate},
\end{align}
where $\bar{s}^{[N]}$ is the observed equilibrium sampled from a graphon game with true parameter $\bar{\eta}$, $\bar{s}_\eta \in L^2([0,1])$ is the equilibrium of the graphon game with parameter $\eta$ and $\Xi \subseteq \mathbb{R}^n$ is the set of admissible parameter values. 
To guarantee uniqueness of $\bar{s}_\eta$, we make the following assumption.

\begin{assumption}[Existence and uniqueness for all $\eta\in\Xi$] \label{a:feas_eta}
For all $\eta\in\Xi$, the graphon game with parameter $\eta$ satisfies Assumption \ref{a:exist_uniq}.
\end{assumption}

We study the performance of the estimator suggested in~\eqref{eta_estimate}, under the following assumption.

\let\theassumption\origtheassumption
\setcounter{assumption}{\the\numexpr\value{assumption}-1}

\begin{assumption}[Identifiability] \label{a:identifiability} 
Suppose that Assumption \ref{a:feas_eta} holds.
The true parameter $\bar{\eta} \in \Xi$ is \textit{identifiable}, that is, there exists $ L_{\bar{\eta}}>0$ such that
\begin{align}
    \| \bar{\eta} - \eta \| \le L_{\bar{\eta}} \| \bar{s}_{\bar{\eta}} - \bar{s}_{\eta} \|_{L^2} \label{identifiability} \quad \forall \eta \in \Xi.
\end{align}
\end{assumption}

Intuitively, the identifiability assumption is needed because if two arbitrarily close graphon equilibria could  be generated by two significantly different parameters, then it would be impossible to identify $\bar{\eta}$ from a single observation of a sampled equilibrium.
It follows immediately from Proposition~\ref{prop:convergence} that, under Assumption~\ref{a:identifiability}, the estimator defined in (\ref{eta_estimate}) is asymptotically consistent in the limit of infinite population.

\begin{proposition}[{\cite[Corollary 1]{parise2019graphon}}] \label{prop:estimate_conv}
Suppose that Assumptions \ref{a:feas_eta} and \ref{a:identifiability} hold.
Then,
\begin{align*}
    \| \hat{\eta} - \bar{\eta} \| \ \overset{\textrm{a.s.}}{\to}\  0 \textrm{ as } N \to \infty.
\end{align*}
\end{proposition}

Overall, the results detailed so far provide a procedure for estimating unknown parameters of sampled network games under two key assumptions. First, one needs to be able to solve the optimization problem in \eqref{eta_estimate}. Second, one needs to be able to verify parameter identifiability as defined in Assumption \ref{a:identifiability}. In the rest of the paper, we investigate these two points. Specifically, in Section \ref{guarantees}, we study properties of the optimization problem in \eqref{eta_estimate}, guaranteeing for example uniqueness of the solution. In Section \ref{LQ_games_unconstr}, we instead investigate parameter identifiability for common LQ network games. 

\section{Parameter estimation properties} \label{guarantees}

In this section, we present results on the smoothness and local convexity of problem \eqref{eta_estimate}.
To this end, we make the following additional assumptions guaranteeing local convexity and smoothness of the graphon equilibrium with respect to parameter variations.

\begin{assumption}[Convex parameter set] \label{a:convex_xi}
The parameter set $\Xi \in \mathbb{R}^n$ is a convex set and contains the true parameter $\bar{\eta}$ in its interior. 
\end{assumption}

\begin{assumption}[Smoothness of equilibrium] \label{a:smooth}
 Under Assumption \ref{a:feas_eta}, the equilibrium $\bar{s}_\eta(x)$ is twice Lipschitz continuously differentiable in $\eta$, uniformly in $x$.
\end{assumption}

While this assumption may seem restrictive, we demonstrate in Section \ref{LQ_games_unconstr} that it holds for various classes of LQ games.
Under the above assumptions, we next derive our main theorem on regularity properties of the objective function in \eqref{eta_estimate}, which we denote by
\begin{align}
    J(\eta) := \| \bar{s}^{[N]} - \bar{s}_\eta \|_{L^2}^2. \label{cost_function}
\end{align}

\begin{theorem} \label{thm:uniqueness} 
Suppose that Assumptions \ref{a:feas_eta}, \ref{a:identifiability}, \ref{a:convex_xi} and \ref{a:smooth} hold. 
Then, 
\begin{enumerate} 
\item $J(\eta)$ is $L$-smooth.
\end{enumerate}
Moreover, for all $\delta > 0$, there exists $\bar{N}>0$ such that with probability $1-\delta$, for all $N>\bar{N}$,
\begin{enumerate}\addtocounter{enumi}{1}
    \item the function $J(\eta)$ in \eqref{cost_function} is locally strictly convex around the true parameter $\bar{\eta}$ and 
    \item the optimization problem (\ref{eta_estimate}) has a globally unique solution.
\end{enumerate}
\end{theorem}
\begin{proof}
1) To prove that $J(\eta)$ is $L$-smooth, note that
\begin{align*}
    J(\eta) &= \| \bar{s}^{[N]} - \bar{s}_\eta \|_{L^2}^2
    = \int_0^1 \left ( \bar{s}^{[N]}(x) - \bar{s}_\eta(x) \right )^2 dx,\\
    \nabla_\eta J(\eta) &= - \int_0^1 2 ( \bar{s}^{[N]}(x) - \bar{s}_\eta(x) )  \nabla_\eta  \bar{s}_\eta(x) dx.
\end{align*} 
By Assumption \ref{a:smooth}, there exists $L_1,L_2 > 0$ such that $\|\bar{s}_\eta - \bar{s}_{\tilde{\eta}} \|_{L^2} \le L_1 \|\eta- \tilde{\eta}\|$ and $\|\nabla_\eta \bar{s}_\eta - \nabla_\eta \bar{s}_{\tilde{\eta}} \|_{L^2;\mathbb{R}^n} \le L_2 \|\eta- \tilde{\eta}\|$.
Hence, 
\begingroup
\allowdisplaybreaks
\begin{align*}
    \frac{1}{2} \| &\nabla_\eta  J(\eta) - \nabla_\eta J(\tilde{\eta}) \| \\ 
    &= \left \|  \int_0^1 \left[ (\bar{s}^{[N]}(x)-\bar{s}_{\eta}(x)) \nabla_\eta \bar{s}_{\eta}(x) \right. \right.\\
    &\quad \left. \left. -  (\bar{s}^{[N]}(x)-\bar{s}_{\tilde{\eta}}(x)) \nabla_\eta \bar{s}_{\tilde{\eta}}(x)  \right] dx \right \| \\
    &\le  \left \|  \int_0^1  (\bar{s}^{[N]}(x) -\bar{s}_\eta(x) )(\nabla_\eta \bar{s}_{\eta}(x) - \nabla_\eta \bar{s}_{\tilde{\eta}}(x)) dx \right\| \\
    &\quad  +\left\| \int_0^1 (\bar{s}_{\tilde{\eta}}(x) - \bar{s}_{\eta}(x))  \nabla_\eta \bar{s}_{\tilde{\eta}}(x) dx  \right \| \\
    &\le \| \bar{s}^{[N]} -\bar{s}_\eta \|_{L^2}  \left \| \nabla_\eta \bar{s}_{\eta} - \nabla_\eta \bar{s}_{\tilde{\eta}} \right \|_{L^2;\mathbb{R}^n} \\
    &\quad + \| \bar{s}_{\tilde{\eta}} - \bar{s}_{\eta}\|_{L^2} \left\|  \nabla_\eta \bar{s}_{\tilde{\eta}} \right \|_{L^2;\mathbb{R}^n} \tag*{\{\small By Cauchy-Schwartz\}} \\
    &\le 2 s_{\rm max} L_2 \|\eta - \tilde{\eta}\| + L_1 \| \eta - \tilde{\eta}\| \cdot L_1 \tag*{\{\small By Lemma \ref{lem:bded_grad_1} in Appendix A\}}\\
    &= \left ( 2 L_2 s_{\rm max}  +  L_1^2  \right ) \| \eta - \tilde{\eta}\|
    =: L_{J} \| \eta - \tilde{\eta}\|.
\end{align*}

2) To investigate the local convexity of $J(\eta)$ around $\bar{\eta}$, we compute the Hessian $H(\eta):= \nabla_\eta^2  J(\eta)$ and examine under which conditions it is positive definite 
    \begin{align}\label{eq:H}
    H(\eta) &  := \nabla_\eta^2 J(\eta) = 
     2 \underbrace{\int_0^1  \nabla_\eta \bar{s}_\eta(x) \nabla_\eta \bar{s}_\eta(x) ^T dx}_{=: T_1(\eta)}\\ \notag
    &\qquad \qquad \quad  - 2 \underbrace{\int_0^1 ( \bar{s}^{[N]}(x) - \bar{s}_\eta(x) ) \nabla_\eta^2 \bar{s}_\eta(x) dx}_{=: T_2^N(\eta)}.
\end{align}

 Lemma \ref{lem:T1_eta_bar} in Appendix B shows that $T_1(\bar \eta)$ is positive definite, while Lemma \ref{lem:T2_eta_bar} shows that for all $\delta>0$, $\epsilon>0$, there exists $\bar{N}$ such that for all $N>\bar{N}$, $\|T_2^N(\bar \eta)\|<\epsilon$ with probability $1-\delta$.  
 It follows that for all $\delta>0$, there exists $\bar{N}$ such that for all $N>\bar{N}$, the Hessian $H(\bar{\eta})$ is positive definite with probability $1-\delta$. 
We next show that $H(\eta) > 0$ locally around $\bar{\eta}$.
To this end, note that
\begin{align*}
    \frac{1}{2}H(\eta) = T_1(\bar{\eta}) &+ [T_1(\eta) - T_1(\bar{\eta})]\\ &+ T_2^N(\bar{\eta}) + [T_2^N(\eta) - T_2^N(\bar{\eta})].
\end{align*}
From Lemma \ref{lem:delta_Ts} in Appendix B, the difference terms can be made arbitrary small for $\eta$ close to $\bar{\eta}$, independently of $N$.
Hence, for all $\delta>0$, it follows from Lemmas~\ref{lem:T1_eta_bar}, \ref{lem:T2_eta_bar}, \ref{lem:delta_Ts} that there exists $\mu$ and $\bar{N}$ such that with probability $1-\delta$, the Hessian $H(\eta)$ is positive definite for all $\eta$ satisfying $\|\eta-\bar{\eta}\|<\mu$ and for all $N>\bar{N}$.
This implies that locally around $\bar{\eta}$, $J(\eta)$ is  strictly convex and there is a unique solution to
\begin{align}
    \underset{\eta \in \Xi, \|\eta - \bar{\eta} \| \le \mu}{\textrm{min}} \quad & \| \bar{s}^{[N]} - \bar{s}_\eta \|_{L^2}^2 \label{equiv_optim}.
\end{align}

3) Finally, to prove that (\ref{eta_estimate}) has a unique solution over the entire domain $ \Xi$, let $\tilde{\eta}$ be any generic solution to (\ref{eta_estimate}).
By identifiability, we have
\begin{align*}
    \| \bar{\eta} - \tilde{\eta} \| &\le L_{\bar{\eta}} \| \bar{s}_{\bar{\eta}} - \bar{s}_{\tilde{\eta}} \|_{L^2}\\
    &\le L_{\bar{\eta}} ( \| \bar{s}_{\bar{\eta}} - \bar{s}^{[N]} \|_{L^2} + \| \bar{s}^{[N]}  - \bar{s}_{\tilde{\eta}} \|_{L^2} )\\
    &\le 2 L_{\bar{\eta}} \| \bar{s}_{\bar{\eta}} - \bar{s}^{[N]} \|_{L^2} \tag*{\{By (\ref{eta_estimate})\}}.
\end{align*}

Therefore, by Proposition \ref{prop:convergence}, for a given $\mu$ and $\delta$, there exists $\bar{N}'\ge \bar{N}$ such that for $N>\bar{N}'$, with probability $1-\delta$, any solution $\tilde{\eta}$ to \eqref{eta_estimate} satisfies $\| \bar{\eta} - \tilde{\eta} \| \le \mu$.
In this case, (\ref{eta_estimate}) is equivalent to (\ref{equiv_optim}) and thus has a unique solution.
\endgroup
\end{proof}

The first part of Theorem \ref{thm:uniqueness} is useful as smoothness of the objective function in \eqref{eta_estimate} is a sufficient condition for the convergence to a stationary point of many derivative-free optimization algorithms such as trust region \cite{conn2009global} and finite difference methods \cite{nesterov2017random}.
Additionally, the second and third parts of Theorem \ref{thm:uniqueness} guarantee that if these algorithms start close enough to the optimal solution, they converge to it for large enough $N$ with high probability.
Global convergence and strong convexity remain an open problem.

\section{Identifiability and Smoothness} \label{LQ_games_unconstr}

The results derived above rely on parameter identifiability (Assumption \ref{a:identifiability}) and smoothness  of the equilibrium (Assumption  \ref{a:smooth}).
We next verify that these assumptions hold in games involving both homogeneous and heterogeneous parameters.

To this end, we  focus on linear quadratic (LQ) games in which the payoff $U$ is linear in the network aggregate $z(x)$ and quadratic in the strategy $s(x)$
\begin{align}
    U(s(x),z(x),& \theta_\eta(x)) = - \frac{1}{2} s(x)^2 \nonumber\\
    &\quad + ([\theta_\eta(x)]_1 + [\theta_\eta(x)]_2 z(x)) s(x) \label{eq:lq_payoff}
\end{align}
where the components of the heterogeneity parameter  $\theta_\eta(x)=[[\theta_\eta(x)]_1,[\theta_\eta(x)]_2]^\top  \in \Theta \subseteq  \mathbb{R}^2_+$ denote the standalone marginal return ($[\theta_\eta(x)]_1$) and the local aggregate effect on marginal return ($[\theta_\eta(x)]_2]$), respectively. 

If the game has homogeneous parameter $\theta_\eta(x)\equiv [\eta_1,\eta_2]^\top \in \mathbb{R}^2_+$ for all $x\in[0,1]$, then under Assumption~\ref{a:feas_eta}, the graphon Nash equilibrium $\bar{s}_\eta$ can be explicitly written as a fixed point of the best-response mapping 
\begin{align} \label{lq_nash_eq_proj}
 \bar{s}_\eta(x) = \Pi_{\mathcal{S}} \left [ \eta_{1} \mathds{1}(x) +  \eta_{2} \mathbb{W} \Bar{s}_\eta(x) \right ]
\end{align}
where $\Pi_{\mathcal{S}} [\cdot]$ is the operator for the projection onto the strategy set $\mathcal{S}$.
Since under this projection operation, different parameters $\eta$ could yield the same equilibrium, we introduce an additional assumption to guarantee identifiability. \begin{assumption}[Internal equilibrium] \label{a:internality}
Under Assumption~\ref{a:feas_eta}, for all $\eta \in \Xi$, the equilibrium $\bar{s}_\eta(x)$ is interior (i.e., $\bar{s}_\eta(x) \in \textrm{int}(\mathcal{S})$) for all $x$. 
\end{assumption}
With this additional assumption, the graphon Nash equilibrium in \eqref{lq_nash_eq_proj} simplifies to
\begin{align} \label{lq_nash_eq}
 \bar{s}_\eta(x) = (\mathbb{I} - \eta_{2} \mathbb{W})^{-1} \eta_{1} \mathds{1}(x),
\end{align}
which corresponds to the Bonacich centrality of agent $x$ in the graphon $W$ \cite{avella2018centrality}. 

\subsection{LQ games with unknown homogeneous parameters} \label{sec:LQ_hom_param}

We start our analysis by proving identifiability and smoothness in homogeneous LQ games when both the parameter $\eta_2$ representing the effect of the local aggregate and the standalone marginal return parameter $\eta_1$ are unknown.

\begin{proposition}[Identifiability] \label{lem:homogeneous}
Consider a LQ game with $\Xi \subset \mathbb{R}_+^2$ and payoff function
\begin{equation}
    U(s(x), z(x), \theta_{{\eta}}(x)) = -\frac{1}{2}s^2(x) + ({\eta}_1 + {\eta}_2 z(x)) s(x) \label{lin_game_two_param},
\end{equation}
so that $\theta_{{\eta}}(x)={\eta} = [{\eta}_1,{\eta}_2]^\top$ for all $x\in[0,1]$.
Suppose that Assumptions \ref{a:feas_eta} and \ref{a:internality} hold.
The parameter vector ${\eta} \in \Xi$ is identifiable if and only if agents have heterogeneous effects on the local aggregate at equilibrium (i.e. $ \bar{z}_{{\eta}} \ne \gamma \mathds{1}$ for any $\gamma \in \mathbb{R}$). 
Otherwise, only the sum ${\eta}_1 + \gamma {\eta}_2$ is identifiable.
\end{proposition}

\begin{proposition}[Smoothness] \label{lem:smooth_lq_hom}
Consider a homogeneous LQ game satisfying all the assumptions in Proposition \ref{lem:homogeneous} and further assume that $\Xi$ is compact and $\eta_2 \|\mathbb{W}\|_\infty <1$ for all $\eta \in \Xi$. 
The corresponding graphon equilibrium satisfies Assumption \ref{a:smooth}.
\end{proposition}

To prove smoothness of the equilibrium in Proposition~\ref{lem:smooth_lq_hom}, we use the assumption that $\eta_2 \| \mathbb{W}\|_\infty < 1$. 
This is a stronger assumption than what is needed for existence and uniqueness of the equilibrium since $\lambda_{\textup{max}}(\mathbb{W}) < \| \mathbb{W}\|_\infty$ for symmetric graphons. The assumption on $\|\mathbb{W}\|_\infty$ (which has an interpretation in terms of max degree of the agents in the graphon)  is required to guarantee  Lipschitz continuity of the equilibrium with respect to the parameter $\eta$, point-wise in $x$. 

\subsection{LQ games with heterogeneous parameters}

The next example generalizes the setting of Section~\ref{sec:LQ_hom_param} by considering parameters which can be heterogeneous across communities of agents.
Specifically, we consider a setting in which agents are partitioned into $K$ communities (with probability $\pi_k$ such that $\sum_{k=1}^K \pi_k = 1$) and connect with a probability depending on the community they belong to.
This random network model (which is essentially a stochastic block model (SBM)) can be captured with a graphon by partitioning $[0,1]$ into $K$ disjoint intervals $\{\mathcal{C}_k\}_{k=1}^K$, each of length $ \pi_k$, and by defining a piecewise constant graphon $W_{\rm SBM}$ as 
\begin{align}
    W_{\rm SBM}(x,y) = Q_{ij} \qquad &\textrm{for all } x\in \mathcal{C}_i,\  y \in \mathcal{C}_j \label{sbm_W}
\end{align}
where $Q_{ij}$ is the probability that an agent from community $i$ interacts with an agent from community $j$ and the corresponding matrix $Q\in\mathbb{R}^{K\times K}$ is symmetric.

We assume that the local aggregate effect is the same for each agent belonging to the same community but is a priori unknown.
In other words, the payoff for each agent $x\in \mathcal{C}_k $ in community $k$ is 
\begin{equation}
    U(s(x),z(x), \theta_{\eta}(x)) = -\frac{1}{2}s^2(x) + (\theta_1 + \eta_k z(x)) s(x) \label{lin_game_SBM}
\end{equation}
where  $\theta_{\eta}(x)=\eta_k$ for each $x\in \mathcal{C}_k$. For simplicity, we assume that the standalone marginal return $ \theta_1 >0 $ is known and homogeneous across agents. 
The parameter to identify is then $\eta= [\eta_1, \hdots ,\eta_K]^\top \in \mathbb{R}^K_+$.

\begin{proposition}[Identifiability] \label{lem:sbm_id}
Consider a LQ game with SBM graphon as defined in \eqref{sbm_W} and payoff as defined in (\ref{lin_game_SBM}) with $\theta_1 >0$.
Let $\Xi \subset \mathbb{R}_+^K$ and suppose Assumptions~\ref{a:feas_eta} and \ref{a:internality} hold.
If $\mathcal{S}\subseteq \mathbb{R}_+$, then the parameter $\eta \in \Xi$ is identifiable. 
\end{proposition}

\begin{proposition}[Smoothness] \label{lem:smooth_lq_sbm}
Consider a heterogeneous LQ game satisfying all the assumptions in Proposition \ref{lem:sbm_id} and further assume that $\Xi$ is compact.
Then, the corresponding graphon equilibrium satisfies Assumption \ref{a:smooth}.
\end{proposition}

\section{Simulations} \label{sec:sims} 
In this section, we provide numerical simulations demonstrating the convergence of the estimator proposed in \eqref{eta_estimate} to the true parameter for the LQ game defined in Proposition~\ref{lem:sbm_id}.
The simulation considers the LQ game with $\bar{\eta}=[0.8,0.6,1,0.8]^\top$, a SBM graphon with 
$$Q = \begin{bmatrix}
    0.9& 0.05& 0& 0\\
    0.05& 0.2& 0.05& 0\\
    0& 0.05& 0.2& 0.05\\
    0& 0& 0.05& 0.8
\end{bmatrix}$$
and equally sized communities ($\pi_k=0.25$). 
Figure \ref{fig:het} illustrates convergence of the estimator $\hat{\eta}$ to $\bar{\eta}$ for large enough $N$ for this example.
The estimation problem \eqref{eta_estimate} was solved using MATLAB's fmincon solver (which uses an interior-point method algorithm).

\begin{figure}
    \centering
    \includegraphics[width=\columnwidth,trim=2.2cm 0 2.2cm 0,clip]{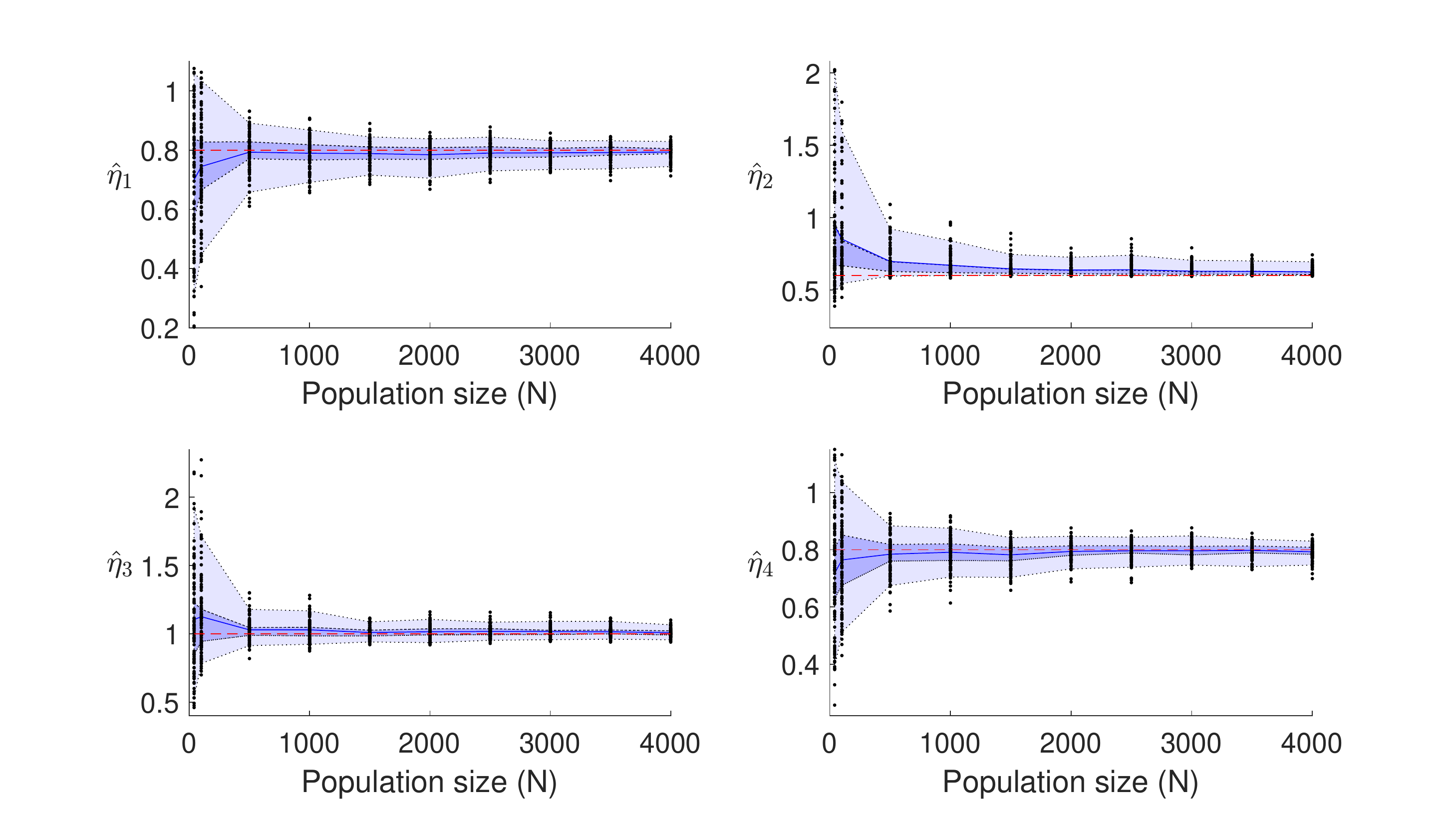}
    \caption{Convergence of estimator. The red dashed line is the true parameter value and quantiles of the estimator values for 100 experiments per each $N$ are shown in blue.}
    \label{fig:het}
\end{figure}

\section{Conclusions} \label{conclusion}

In this work, we introduced a method for estimation of  payoff parameters in large network games by leveraging the framework of graphon games. We discussed properties of the corresponding optimization problem and proved parameter identifiability for linear quadratic games with both homogeneous and heterogeneous parameters. Identifiability of parameters for  games with nonlinear dependence on the network aggregate such as the \textit{quadratic quadratic game} presented in \cite{parise2019graphon} is a future research direction. 


\section*{APPENDIX}
\begingroup
\allowdisplaybreaks
\subsection{Auxiliary lemmas}

Lemmas  \ref{lem:bded_grad_1} and \ref{lem:bded_grad_2} are standard results. 

\begin{lemma} \label{lem:bded_grad_1}
Assumption \ref{a:smooth} holds if and only if there exists $L_1, L_2>0$ such that
\begin{align*}
    \left | \frac{\partial \bar{s}_\eta(x)}{\partial \eta_i} \right| &\le \| \nabla_\eta \bar{s}_\eta (x) \| \le L_1 \quad \textrm{for all } x \in[0,1]\textrm{ and}\\
    \left | \frac{\partial^2 \bar{s}_\eta(x)}{\partial \eta_i \partial \eta_j} \right| &\le \| \nabla_\eta^2 \bar{s}_\eta (x) \| \le L_2 \quad \textrm{for all } x\in[0,1].
\end{align*}
\end{lemma}

\begin{lemma} \label{lem:bded_grad_2} 
Take $f:\mathbb{R}^K \to \mathbb{R}$ continuous and differentiable. 
If $\| \nabla_\eta f(\eta) \| \le L$ for all $\eta$ , then $f$ is $L$-Lipschitz continuous.
\end{lemma}

\begin{lemma} \label{lem:series}
Consider a series of the form
\begin{align*}
    f_h(\eta,x) := \sum_{k=h}^\infty k(k-1) \dots (k-h+1) \eta^{k-h} \alpha_k(x)
\end{align*}
for $\alpha_k \in L^2([0,1])$, $h\in\mathbb{N}$, $\eta \in \Xi\subset \mathbb{R}$ and $x \in [0,1]$.  Suppose that: i) $\exists\, \eta_{\rm max}$ such that $| \eta| \le \eta_{\rm max}$ for all $\eta\in\Xi$, ii) $\exists\, \beta >0$ such that $|\alpha_k(x)| \le \beta^k, \forall x\in[0,1], \forall k\in\mathbb{N}$ and iii) $\eta_{\rm max} \beta < 1$. Then the series $f_h(\eta,x)$ converges and is $L_h$-Lipschitz continuous in $\eta$ uniformly in $x$.
\end{lemma}

\begin{proof} 
First, we show that $f_h(\eta,x)$ converges by comparison test with a convergent geometric series.
Note that
\begingroup
\allowdisplaybreaks
\begin{align*}
    \textstyle  \sum_{k=h}^\infty & \textstyle  \left | k(k-1) \dots (k-h+1) \eta^{k-h} \alpha_k(x) \right | \\
    &\textstyle  \le \beta^h \sum_{k=h}^\infty k(k-1) \dots (k-h+1) ( |\eta| \beta)^{k-h}\\
    &\textstyle \overset{(*)}{=} \beta^h \frac{d^h}{dz^h}\left(\frac{1}{1-z}\middle) \right|_{z = |\eta| \beta} 
    = \beta^h \frac{h!}{(1-|\eta|\beta)^{h+1}}\\
    &\textstyle \le \beta^h \frac{h!}{(1-\eta_{\rm max} \beta)^{h+1}}
    =: B_{h}
\end{align*}
\endgroup
where $(*)$ follows from \cite{andrews1998geometric} since $|\eta| \beta < 1$.
Hence, by the comparison test for series, $f_h(\eta,x)$ converges and $| f_h(\eta,x)| \le B_h$ for all $\eta\in\Xi$ and $x\in[0,1]$. 
Moreover, for all $h$, $f_h(\eta,x)$ is Lipschitz continuous in $\eta$ uniformly in $x$ since, if we denote $a_{h,k}:=k(k-1)\dots(k-h+1)$,
\begin{align*}
     \textstyle  |f_h(\eta&,x) - f_h (\tilde{\eta},x)|
     \textstyle  \le \left | \sum_{k=h}^\infty a_{h,k} ( \eta^{k-h} - \tilde{\eta}^{k-h} ) \alpha_k(x) \right |\\
    &\textstyle  \overset{(*)}{\le} \sum_{k=h+1}^\infty a_{h,k} | \eta - \tilde{\eta} | (k-h) ( \eta_{\rm max})^{k-h-1} | \alpha_k(x)  |\\
    &\textstyle  \le B_{h+1} |\eta - \tilde{\eta} | 
    =: L_h |\eta - \tilde{\eta} | \quad \forall x \in [0,1]
\end{align*}
where $(*)$ holds by the identity \cite[p.83]{selby1967standard} $x^{p+1}-y^{p+1}=(x-y)(x^p+x^{p-1}y+\dots+xy^{p-1}+y^p)$ and $|\eta|, |\tilde{\eta}| \le \eta_{\rm max}$.
\end{proof}

\begin{lemma} \label{lem:sbm}
Consider the Nash equilibrium $\bar{s}_\eta$ of a LQ graphon game with graphon $W_{\rm SBM}$ defined as in \eqref{sbm_W}, payoff defined as in \eqref{lin_game_SBM}, with $\eta \in \mathbb{R}_+^K$ and satisfying Assumption \ref{a:internality}.
Then, $\bar{s}_\eta(x) = \bar{\bar{s}}_\eta^k$ for all $x \in \mathcal{C}_k$
where $\bar{\bar{s}}_\eta \in\mathbb{R}^K$ can be characterized as
\begin{align*}
    \bar{\bar{s}}_\eta = \theta_1 ( I - \Delta_\eta Q \Delta_\pi)^{-1} \mathds{1}
\end{align*}
where $\Delta_\eta := \textrm{diag}(\eta_1, ..., \eta_K)$ and $\Delta_\pi := \textrm{diag}(\pi_1,...,\pi_K)$.
Additionally, 
\begin{align*}
    \|\bar{s}_\eta - \bar{s}_{\bar{\eta}}\|_{L^2} \ge \sqrt{\min_k(\pi_k)} \| \bar{\bar{s}}_\eta - \bar{\bar{s}}_{\bar{\eta}} \| .
\end{align*}
\end{lemma}

\begin{proof}
    The first part of the statement follows from a generalization of \cite[Proposition 1]{avella2018centrality}.
    Then,
    $
    \| \bar{s}_{\eta} - \bar{s}_{\bar{\eta}} \|_{L^2}^2 =  \sum_{k=1}^K \int_{\mathcal{C}_k} (\bar{\bar{s}}_{\eta}^k - \bar{\bar{s}}_{\bar{\eta}}^k)^2 dx
    = \sum_{k=1}^K \pi_k (\bar{\bar{s}}_{\eta}^k - \bar{\bar{s}}_{\bar{\eta}}^k)^2 
    = (\bar{\bar{s}}_{\eta} - \bar{\bar{s}}_{\bar{\eta}})^T \Delta_\pi (\bar{\bar{s}}_{\eta} - \bar{\bar{s}}_{\bar{\eta}})  
    \ge \min_k(\pi_k) \|\bar{\bar{s}}_{\eta} - \bar{\bar{s}}_{\bar{\eta}}\|^2.
    $
\end{proof}

\subsection{Auxiliary lemmas for the proof of Theorem \ref{thm:uniqueness}}

\begin{lemma} \label{lem:T1_eta_bar}
Suppose that Assumptions \ref{a:feas_eta}, \ref{a:identifiability}, \ref{a:convex_xi} and \ref{a:smooth} hold.
Then, $T_1(\bar{\eta})$ as defined in \eqref{eq:H} is positive definite.
\end{lemma}

\begin{proof}
Showing that $T_1(\bar{\eta})$ is positive definite 
requires that for any nonzero $v \in \mathbb{R}^n$,
\begin{align}
    & v^T T_1(\bar{\eta}) v > 0 \Leftrightarrow  \int_0^1 ( v^T \nabla_\eta \bar{s}_{\bar{\eta}}(x) )^2 dx > 0. \label{eq:pd_T1}
\end{align}
Suppose that there exists $v \ne 0$ such that $v^T \nabla_\eta \bar{s}_\eta(x) = 0$ \textit{almost everywhere in $x$} (so that the integral in \eqref{eq:pd_T1} would be zero).
Let $\tilde{\eta} = \bar{\eta} + v$ be a perturbation of the true parameter $\bar{\eta}$, where, without loss of generality, we assume that the perturbation is small enough to ensure that $\bar{\eta} + v \in \textrm{int}(\Xi)$ and that $\|v\| < \frac{2}{L_{\bar{\eta}} L_2} $ where $L_{\bar{\eta}}$ is as defined in Assumption~\ref{a:identifiability} and $L_2$ is as defined in Lemma~\ref{lem:bded_grad_1} (by Assumption \ref{a:smooth}). 
Such a perturbation $v$ must exist since $\bar{\eta}$ is in the interior of~$\Xi$ (by Assumption \ref{a:internality}). 
This perturbed parameter $\tilde{\eta}$ results in a new equilibrium $\bar{s}_{\tilde{\eta}} = \bar{s}_{\bar{\eta} + v}$.
By Taylor's expansion for multivariate functions \cite[Theorem 12.14]{apostol1974analysis}, for almost every~$x$, we have
\begin{align*}
    \bar{s}_{\tilde{\eta}}(x) &= \bar{s}_{\bar{\eta}}(x) + \overbrace{(\tilde{\eta} - \bar{\eta})^T \nabla_\eta \bar{s}_{\bar{\eta}} (x) }^{=v^T\nabla_\eta \Bar{s}_{\Bar{\eta}}(x)=0}\\
    &\quad + \frac{1}{2} (\tilde{\eta} - \bar{\eta})^T \Big(  \nabla_{\eta}^2 \bar{s}_\eta (x) \Big |_{\eta = \bar{\eta} + \zeta_x v} \Big) (\tilde{\eta} - \bar{\eta})
\end{align*}
for some $\zeta_x \in[0,1]$. 
Note that $\bar{\eta} + \zeta_x v \in \textrm{int}(\Xi)$ since it is a convex combination of $\bar{\eta}$ and $\tilde{\eta}$, and $\Xi$ is a convex set.
By Lemma \ref{lem:bded_grad_1}, this implies that almost everywhere in $x$,
\begin{align*}
    | \bar{s}_{\tilde{\eta}}(x) - \bar{s}_{\bar{\eta}}(x)| &= \frac{1}{2} \left| v^T \Big ( \nabla_{\eta}^2 \bar{s}_\eta (x) \Big |_{\eta = \bar{\eta} + \zeta_x  v} \Big) v \right| \le \frac{1}{2} L_2 \|v\|^2.
\end{align*}
Then, $\| \bar{s}_{\tilde{\eta}} - \bar{s}_{\eta} \|_{L^2} \le \frac{1}{2} L_2 \|v\|^2$.
Identifiability of $\bar{\eta}$ yields 
\begin{align*}
     \| \tilde{\eta} - \bar{\eta} \|  \le L_{\bar{\eta}} \| \bar{s}_{\tilde{\eta}} - \bar{s}_{\bar{\eta}} \|_{L^2} 
    \Leftrightarrow \| \bar{s}_{\tilde{\eta}} - \bar{s}_{\bar{\eta}} \|_{L^2}  \ge \frac{ \| v\|}{L_{\bar{\eta}}}. 
\end{align*}
Together, these two inequalities imply that
$
    \|v\|/L_{\bar{\eta}} \le \frac{1}{2} L_2 \|v\|^2 
    \Leftrightarrow \|v\| \ge \frac{2}{L_{\bar{\eta}} L_2}
$
which contradicts our choice of $v$ such that $ \| v\| <  \frac{2}{L_{\bar{\eta}} L_2} $.
\end{proof}

\begin{lemma} \label{lem:T2_eta_bar}
Suppose that Assumptions \ref{a:feas_eta} and \ref{a:smooth} hold. 
Then, $T_2^N(\bar{\eta})$ as defined in \eqref{eq:H} converges to the zero matrix in probability, i.e., for all $\delta > 0$ and for all $\epsilon_1 > 0$, there exists $\bar{N}$ such that if $N > \bar{N}$, 
\begin{align*}
    \mathbb{P} \left ( \left [T_2^N (\bar{\eta}) \right]_{ij} \le \epsilon_1 \quad  \forall i,j \right ) \ge 1 - \delta.
\end{align*}
\end{lemma}

\begin{proof}
Consider the $ij$th element of $T_2^N(\bar{\eta})$. 
By Hölder's inequality and Lemma \ref{lem:bded_grad_1}, there exists $L_2 >0 $ such that
\begin{align*}
     &[T_2^N(\bar{\eta})]_{ij} = \int_0^1 ( \bar{s}^{[N]}(x) - \bar{s}_{\bar{\eta}}(x) ) \frac{\partial^2 \bar{s}_{\bar{\eta}}(x)}{\partial \eta_i \eta_j} dx\\
     & \le \left( \int_0^1 ( \bar{s}^{[N]}(x) - \bar{s}_{\bar{\eta}}(x) )^2 dx \right)^{\frac{1}{2}}  \left( \int_0^1 \left( \frac{\partial^2 \bar{s}_{\bar{\eta}}(x)}{\partial \eta_i \eta_j} \right)^2 dx\right)^{\frac{1}{2}}\\
     & \le L_2 \| \bar{s}^{[N]} - \bar{s}_{\bar{\eta}} \|_{L^2} .
\end{align*}

By Proposition \ref{prop:convergence}, $\| \bar{s}^{[N]} - \bar{s}_{\bar{\eta}} \|_{L^2} \overset{\textrm{a.s.}}{\to} 0$ as $N\to\infty$ which implies that, for all $\delta>0$, $\| \bar{s}^{[N]} - \bar{s}_{\bar{\eta}} \|_{L^2} \to 0$ with probability $1-\delta$.
Therefore, for all $\delta>0$, and all $\epsilon_1>0$, there exists some $\bar{N}$ such that for $N > \bar{N}$,
\begin{align*}
    [T_2^N(\bar{\eta})]_{ij}  \le L_2 \| \bar{s}^{[N]} - \bar{s}_{\bar{\eta}} \|_{L^2} \le \epsilon_1 \quad \textrm{w.p. } 1-\delta. \quad \quad \qquad \qedhere
\end{align*}
\end{proof}

\begin{lemma} \label{lem:delta_Ts}
Suppose that Assumption \ref{a:smooth} holds. Then,
\begin{enumerate}[label=(\roman*)]
    \item $T_1(\eta)$ in \eqref{eq:H} is continuous in $\eta$, i.e. for all $\epsilon_2>0$, there exists $\mu_1>0$ such that if $\|\eta - \bar{\eta}\| \le \mu_1$, then $\Big | [T_1(\eta)]_{ij} - [T_1(\bar{\eta})]_{ij} \Big | \le \epsilon_2 $ for all $i,j$, and
    \item $T_2^N(\eta)$ in \eqref{eq:H} is continuous in $\eta$ uniformly in $N$, i.e., for all $\epsilon_3>0$, there exists $\mu_2>0$ such that if $\|\eta - \bar{\eta}\| \le \mu_2$, then $\Big | [T_2^N(\eta)]_{ij} - [T_2^N(\bar{\eta})]_{ij} \Big | \le \epsilon_3$ for all $i,j, N.$
\end{enumerate}
\end{lemma}
\begin{proof}
(i) Consider $T_1(\eta) = \int_0^1 \nabla_\eta \bar{s}_\eta(x)   \nabla_\eta \bar{s}_\eta(x) ^T dx$. Since $\nabla_\eta \bar{s}_\eta(x)$ is continuous in $\eta$ by Assumption \ref{a:smooth}, (i) holds.
\\ \\
(ii) Consider $T_2^N(\eta) = \int_0^1 ( \bar{s}^{[N]}(x) - \bar{s}_\eta(x) ) \nabla_\eta^2 \bar{s}_\eta(x) dx$.
By Assumption \ref{a:smooth}, for any $\epsilon_3>0$, there exists $ \mu_2$ such that if $\|\eta -\bar{\eta}\| < \mu_2$, 
\begin{align*}
    \left | \bar{s}_\eta(x) - \bar{s}_{\bar{\eta}}(x) \right| &\le \frac{1}{3 L_2}\epsilon_3, \\
    \left | \frac{\partial^2 \bar{s}_\eta(x)}{\partial \eta_i \partial \eta_j} - \frac{\partial^2 \bar{s}_{\bar{\eta}}(x)}{\partial \eta_i \partial \eta_j} \right| &\le \frac{1}{3 s_{\rm max}} \epsilon_3 \qquad \forall i,j \in\{ 1,...,n\}
\end{align*}
for any $x\in[0,1]$ where $L_2$ is as in Lemma \ref{lem:bded_grad_1}.
Then, for $\|\eta - \bar{\eta}\| \le \mu_2$, it holds that
\begingroup
\allowdisplaybreaks
\begin{align*}
    \Big | [T_2^N(\eta)&]_{ij} - [T_2^N(\bar{\eta})]_{ij} \Big | \le \left | \int_0^1 ( \bar{s}^{[N]}(x) - \bar{s}_\eta(x) ) \frac{\partial^2 \bar{s}_\eta(x)}{\partial\eta_i\partial\eta_j} \right.\\
    &\qquad \left. - ( \bar{s}^{[N]}(x) - \bar{s}_{\bar{\eta}}(x) ) \frac{\partial^2 \bar{s}_{\bar{\eta}}(x)}{\partial\eta_i\partial\eta_j} dx \right |\\
    &\le \int_0^1 \left |  \bar{s}^{[N]}(x) \right|  \left | \frac{\partial^2 \bar{s}_\eta(x)}{\partial\eta_i\partial\eta_j} - \frac{\partial^2 \bar{s}_{\bar{\eta}}(x)}{\partial\eta_i\partial\eta_j}\right|  \\
    &\qquad + \left| \bar{s}_{\bar{\eta}}(x) - \bar{s}_{\eta}(x) \right| \left|  \frac{\partial^2 \bar{s}_{\bar{\eta}}(x)}{\partial\eta_i\partial\eta_j} \right|\\
    &\qquad + | \bar{s}_\eta(x)|  \left| \frac{\partial^2 \bar{s}_{\bar{\eta}}(x)}{\partial\eta_i\partial\eta_j}  -  \frac{\partial^2 \bar{s}_\eta(x)}{\partial\eta_i\partial\eta_j} \right| dx\\
    &\le \int_0^1 s_{\rm max}  \frac{\epsilon_3}{3s_{\rm max}}   + \frac{\epsilon_3}{3L_2} L_2 + s_{\rm max} \frac{\epsilon_3}{3 s_{\rm max}}  dx = \epsilon_3
\end{align*}
\endgroup
where the last inequality holds by Lemma \ref{lem:bded_grad_1}.
\end{proof}

\subsection{Omitted proofs}
\noindent \textbf{Proof of Proposition \ref{lem:homogeneous}}\\
By Assumption \ref{a:internality}, 
it follows from \eqref{lq_nash_eq} that, for any $\eta, \bar{\eta} \in \Xi$, we have
$
    (\mathbb{I} - \eta_2 \mathbb{W}) \bar{s}_\eta = \eta_1 \mathds{1} \textrm{ and }
    (\mathbb{I} - \bar{\eta}_2 \mathbb{W}) \bar{s}_{\bar{\eta}} = \bar{\eta}_1 \mathds{1}.
$
Subtracting one expression from the other, we get
\begin{align*}
    (\eta_1 &- \bar{\eta}_1) \mathds{1} = (\bar{s}_\eta - \bar{s}_{\bar{\eta}}) - \eta_2 \mathbb{W} \bar{s}_\eta + \bar{\eta}_2 \mathbb{W} \bar{s}_{\bar{\eta}}  \\
    &= (\bar{s}_\eta - \bar{s}_{\bar{\eta}}) - \eta_2 \mathbb{W} (\bar{s}_\eta - \bar{s}_{\bar{\eta}}) + (\bar{\eta}_2 - \eta_2) \mathbb{W} \bar{s}_{\bar{\eta}}  \\ \Leftrightarrow
    (\eta_1 &- \bar{\eta}_1) \mathds{1} + (\eta_2 - \bar{\eta}_2) \mathbb{W} \bar{s}_{\bar{\eta}}  = (\mathbb{I} - \eta_2 \mathbb{W}) (\bar{s}_\eta - \bar{s}_{\bar{\eta}}) .
\end{align*}

Taking the norm of both sides yields
\begin{align}
    \|(\eta_1 &- \bar{\eta}_1) \mathds{1} + (\eta_2 - \bar{\eta}_2) \mathbb{W} \bar{s}_{\bar{\eta}} \|_{L^2} \label{twoparam_upperbound}\\
    &= \| (\mathbb{I} - \eta_2 \mathbb{W}) (\bar{s}_\eta - \bar{s}_{\bar{\eta}}) \|_{L^2} 
    \le \| \mathbb{I} - \eta_2 \mathbb{W} \|  \| \bar{s}_\eta - \bar{s}_{\bar{\eta}} \|_{L^2} \nonumber\\
    &\le (1 + \eta_2 \lambda_{\rm max}(\mathbb{W})) \| \bar{s}_\eta - \bar{s}_{\bar{\eta}} \|_{L^2} 
    < 2 \| \bar{s}_\eta - \bar{s}_{\bar{\eta}} \|_{L^2} \nonumber
\end{align}
where the last inequality holds since by Assumption \ref{a:feas_eta}, $\eta_2 \lambda_{\rm max}(\mathbb{W}) < 1$ for all $\eta \in \Xi$.
Let $\bar{z}_{\bar{\eta}} := \mathbb{W} \bar{s}_{\bar{\eta}}$ be decomposed into two orthogonal components $\bar{z}_{\bar{\eta}} = \bar{z}_{\bar{\eta}}^{\parallel} + \bar{z}_{\bar{\eta}}^\bot$ where $\bar{z}_{\bar{\eta}}^{\parallel} = \gamma \mathds{1}$ for some $\gamma  \in\mathbb{R}$ and $\langle \bar{z}_{\bar{\eta}}^\bot, \mathds{1} \rangle = 0$.
The left hand side of \eqref{twoparam_upperbound} can then be lower bounded as follows 
\begin{align}
    &\|(\eta_1 - \bar{\eta}_1) \mathds{1} + (\eta_2 - \bar{\eta}_2) (\bar{z}_{\bar{\eta}}^{\parallel} + \bar{z}_{\bar{\eta}}^\bot) \|^2_{L^2} \label{twoparam_lowerbound} \\
    &= \|( (\eta_1 - \bar{\eta}_1) + \gamma(\eta_2 - \bar{\eta}_2) )  \mathds{1} + (\eta_2 - \bar{\eta}_2)  \bar{z}_{\bar{\eta}}^\bot \|^2_{L^2} \nonumber\\
    &\ge ( ( (\eta_1 - \bar{\eta}_1) + \gamma(\eta_2 - \bar{\eta}_2) )^2 + (\eta_2 - \bar{\eta}_2)^2) \bar{\nu} \nonumber\\
    &= \begin{bmatrix}
        \eta_1 - \bar{\eta}_1& \eta_2 - \bar{\eta}_2
    \end{bmatrix}
    \begin{bmatrix}
        1 & \gamma \\ \gamma & \gamma^2 + 1
    \end{bmatrix}
    \begin{bmatrix}
        \eta_1 - \bar{\eta}_1\\ \eta_2 - \bar{\eta}_2
    \end{bmatrix} \bar{\nu} \nonumber \\
    &\ge \lambda_{\rm min} \begin{pmatrix}
        1 & \gamma \\ \gamma & \gamma^2 + 1
    \end{pmatrix} \bar{\nu} \begin{Vmatrix}
        \eta_1 - \bar{\eta}_1\\ \eta_2 - \bar{\eta}_2
    \end{Vmatrix}^2 =: \lambda_{m} \bar \nu \|\eta-\bar \eta\|^2\notag
\end{align}
where $\bar{\nu} := \min(1,\|\bar{z}_{\bar{\eta}}^\bot\|_{L^2})$.
We distinguish two cases.

\textbf{Case 1:}
If $\bar{z}_{\bar{\eta}} \ne \gamma \mathds{1}$, it must be that $\bar{z}_{\bar{\eta}}^\bot \ne 0$. 
The minimal eigenvalue $\lambda_{m}= ((\gamma^2+2)-\sqrt{(\gamma^2+2)^2-4})/2 $ in (\ref{twoparam_lowerbound}) is positive for any $\gamma$, 
hence, combining (\ref{twoparam_upperbound}) and (\ref{twoparam_lowerbound}) yields parameter identifiability
\begin{align*}
    \|\eta - \bar{\eta}\| &< 2(\lambda_{m}\bar{\nu})^{-\frac{1}{2}} \| \bar{s}_\eta - \bar{s}_{\bar{\eta}} \|_{L^2}
    =: L_{\bar{\eta}} \| \bar{s}_\eta - \bar{s}_{\bar{\eta}} \|_{L^2}.
\end{align*} 

\textbf{Case 2:}
If $\bar{z}_{\bar{\eta}} = \gamma \mathds{1}$ (i.e., when $\bar{z}_{\bar{\eta}}^\bot = 0$), then the left hand side of \eqref{twoparam_upperbound} becomes 
$
    \|(\eta_1 - \bar{\eta}_1) \mathds{1} + (\eta_2 - \bar{\eta}_2) \mathbb{W} \bar{s}_{\bar{\eta}} \|_{L^2}
   = \|(\eta_1 - \bar{\eta}_1) \mathds{1} + (\eta_2 - \bar{\eta}_2) \gamma \mathds{1} \|_{L^2} 
    = |(\eta_1+\gamma \eta_2) - (\bar{\eta}_1 + \gamma \bar{\eta}_2) | .
$
From \eqref{twoparam_upperbound}, we then obtain
\begin{align*}
    |(\eta_1+\gamma \eta_2) &- (\bar{\eta}_1 + \gamma \bar{\eta}_2) | 
    \le 2 \| \bar{s}_\eta - \bar{s}_{\bar{\eta}} \|_{L^2} \\
    &=: L_{\bar{\eta}} \| \bar{s}_\eta - \bar{s}_{\bar{\eta}} \|_{L^2} .
\end{align*}
Hence $\eta_1 + \gamma \eta_2$ is identifiable.
Finally, to show that when $\bar{z}_\eta = \gamma \mathds{1}$, $\eta$ is not identifiable, we provide a counterexample.
Consider a simple LQ graphon game satisfying Assumption~\ref{a:feas_eta} and \ref{a:internality} with constant graphon $W(x,y) = c$ for some $c\in\mathbb{R}$.  
By \eqref{lq_nash_eq}, for any $\eta$, the unique Nash equilibrium of this LQ game and its corresponding local aggregate are
$
    \bar{s} = \frac{\eta_1}{1-\eta_2 c}\mathds{1} \textrm{ and }
    \bar{z} = c \bar{s} = c \frac{\eta_1}{1-\eta_2 c}\mathds{1}.
$
It is then clear that for any pair of parameters $(\eta_1,\eta_2)$ such that 
$
    \frac{\eta_1}{1 - \eta_2 c} = \frac{\bar{\eta}_1}{1 - \bar{\eta}_2 c},
$
the identifiability condition will be violated since $\begin{bmatrix}\eta_1 - \bar{\eta}_1\\\eta_2 - \bar{\eta}_2\end{bmatrix} \ne 0$ but $\| \bar{s}_{\eta} - \bar{s}_{\bar{\eta}} \|_{L^2} = 0$.

\noindent \textbf{Proof of Proposition \ref{lem:smooth_lq_hom}}

By Assumption \ref{a:feas_eta}, since $\eta_2 \|\mathbb{W}\| < 1$ for all $\eta_2 \in \Xi$, the homogeneous LQ game equilibrium in \eqref{lq_nash_eq} can be rewritten by using the Neumann series
$
    \bar{s}_\eta = (\mathbb{I} - \eta_2 \mathbb{W})^{-1} \eta_1 \mathds{1}= \eta_1  \sum_{k=0}^\infty \eta_2^k \mathbb{W}^k \mathds{1}. 
$
Hence, for each $x$,
\begin{align*}
    \bar{s}_\eta(x) &= \eta_1  \sum_{k=0}^\infty \eta_2^k (\mathbb{W}^k \mathds{1})(x)
    = \eta_1 f_0(\eta_2, x)
\end{align*}
where $f_h(\eta_2,x)$ is as defined in Lemma \ref{lem:series} with $\alpha_k(x) := (\mathbb{W}^k \mathds{1})(x)$. 
We next compute the partial derivatives of $\bar{s}_\eta$ and express them in terms of $f_h(\eta_2,x)$
\begingroup
\allowdisplaybreaks
\begin{align*}
    \frac{\partial \bar{s}_\eta(x)}{\partial \eta_1} &=  f_0(\eta_2,x), \quad \frac{\partial^2 \bar{s}_\eta(x)}{\partial \eta_1^2} = 0,\\
    \frac{\partial \bar{s}_\eta(x)}{\partial \eta_2} &= \eta_1 \sum_{k=1}^\infty k \eta_2^{k-1} \mathbb{W}^k \mathds{1} = \eta_1 f_1(\eta_2,x),\\
    \frac{\partial^2 \bar{s}_\eta(x)}{\partial \eta_2^2} &= \eta_1 \sum_{k=2}^\infty k(k-1) \eta_2^{k-2} \mathbb{W}^k \mathds{1} = \eta_1 f_2(\eta_2,x),\\
    \frac{\partial^2 \bar{s}_\eta(x)}{\partial \eta_1 \partial \eta_2} &= \frac{\partial^2 \bar{s}_\eta(x)}{\partial \eta_2 \partial \eta_1} = \sum_{k=1}^\infty k \eta_2^{k-1} \mathbb{W}^k \mathds{1} = f_1(\eta_2,x).
\end{align*}
\endgroup
Since $\eta_2^{\rm max} <\infty$ exists by compactness of $\Xi$, $|\alpha_k(x)| \le \|\mathbb{W}\|^k_{\infty}$ and $\eta_2^{\rm max} \|\mathbb{W}\|_\infty < 1$ by assumption, Lemma \ref{lem:series} with $\beta=\|\mathbb{W}\|_\infty$ ensures that $f_h(\eta_2,x)$ is well defined, uniformly bounded by a positive constant $B_h$ and $L_h$-Lipschitz continuous in $\eta_2$ uniformly in $x$ for any $h \in \mathbb{N}$.
Hence, the equilibrium and its partial derivatives are well defined.
To prove Lipschitz continuity, it suffices to show that $\eta_1 f_h(\eta_2,x)$ is Lipschitz continuous in $\eta$ for all $x\in[0,1]$ and for $h=0,1,2$.
This holds since
\begin{align*}
    | \eta_1 f_h&(\eta_2,x) - \tilde{\eta}_1 f_h(\tilde{\eta}_2, x) |\\
    &= | (\eta_1 - \tilde{\eta}_1) f_h(\eta_2,x) - \tilde{\eta}_1 (f_h(\tilde{\eta}_2, x) - f_h(\eta_2,x)) |\\
    &\le | \eta_1 - \tilde{\eta}_1| | f_h(\eta_2,x) | + |\tilde{\eta}_1| |f_h(\tilde{\eta}_2, x) - f_h(\eta_2,x) |\\
    &\le | \eta_1 - \tilde{\eta}_1| B_h + \eta_1^{\rm max} L_h |\eta_2 - \tilde{\eta}_2 |\\
    &\le \max(B_h, \eta_1^{\rm max} L_h) (( \eta_1 - \tilde{\eta}_1)^2 + (\eta_2 - \tilde{\eta}_2)^2)\\
    &=: K_h \| \eta - \tilde{\eta} \| \quad \forall x \in [0,1]
\end{align*}
for $K_h := \max(B_h, \eta_1^{\rm max} L_h)$ and $\eta_1^{\rm max} := \max_{\eta \in \Xi} |\eta_1| < \infty$ since $\Xi$ is compact. 

\noindent \textbf{Proof of Proposition \ref{lem:sbm_id}}

By Lemma \ref{lem:sbm}, we can relate the graphon game equilibrium to a vector $\bar{\bar{s}}_{\eta} \in \mathbb{R}^K$ satisfying
\begin{align*}
    \bar{\bar{s}}_{\eta} - \Delta_{\eta} Q \Delta_\pi \bar{\bar{s}}_{\eta}  &= \theta_1 \mathds{1}.
\end{align*}
Subtracting to $\bar{\bar{s}}_{\bar{\eta}}$ the expression for the generic equilibrium $\bar{\bar{s}}_{\eta} \in \mathbb{R}^K$, we get
\begin{align*}
    (\bar{\bar{s}}_{\eta} - \bar{\bar{s}}_{\bar{\eta}}) - \Delta_\eta Q \Delta_\pi \bar{\bar{s}}_{\eta} + \Delta_{\bar{\eta}} Q \Delta_\pi \bar{\bar{s}}_{\bar{\eta}} &= 0\\
    (\bar{\bar{s}}_{\eta} - \bar{\bar{s}}_{\bar{\eta}}) - \Delta_\eta Q \Delta_\pi (\bar{\bar{s}}_{\eta} - \bar{\bar{s}}_{\bar{\eta}}) + (\Delta_{\bar{\eta}} - \Delta_\eta) Q \Delta_\pi \bar{\bar{s}}_{\bar{\eta}} &= 0\\
    (\mathbb{I} - \Delta_\eta Q \Delta_\pi )(\bar{\bar{s}}_{\eta} - \bar{\bar{s}}_{\bar{\eta}})   = (\Delta_\eta - \Delta_{\bar{\eta}}) &Q \Delta_\pi \bar{\bar{s}}_{\bar{\eta}}.
\end{align*}
Then, taking the norm of both sides, we have
\begin{align}
    &\| (\Delta_\eta - \Delta_{\bar{\eta}}) Q \Delta_\pi \bar{\bar{s}}_{\bar{\eta}} \| = \| (I - \Delta_\eta Q \Delta_\pi)(\bar{\bar{s}}_{\eta} - \bar{\bar{s}}_{\bar{\eta}})  \| \label{sbm_bound_eq}\\
    &\le (\| I\| + \|\Delta_\eta\| \cdot \|Q \Delta_\pi \| )  \| \bar{\bar{s}}_{\eta} - \bar{\bar{s}}_{\bar{\eta}} \| \nonumber\\
    &\overset{(a)}{\le} (1 + \max_k(\eta_k)  \lambda_{\rm max}(Q \Delta_\pi)  )  \| \bar{\bar{s}}_{\eta} - \bar{\bar{s}}_{\bar{\eta}} \| \overset{(b)}{\le} 2 \| \bar{\bar{s}}_{\eta} - \bar{\bar{s}}_{\bar{\eta}} \| \nonumber
\end{align} 
where (a) follows from \cite[Lemma 10]{parise2019graphon} since $Q$ is symmetric, and since $\lambda_{\rm max}(Q\Delta_\pi)=\lambda_{\rm max}(\mathbb{W})$, (b) follows by Assumption \ref{a:feas_eta} as $\max_k(\eta_k) \lambda_{\rm max}(Q\Delta_\pi) < 1$.
Let $\bar{\bar{z}}_{\bar{\eta}} := Q \Delta_\pi \bar{\bar{s}}_{\bar{\eta}}$. 
The left hand side of \eqref{sbm_bound_eq} can be lower bounded as follows
\begin{align*}
    \textstyle \| (&\Delta_\eta -  \Delta_{\bar{\eta}}) \bar{\bar{z}}_{\bar{\eta}}  \|^2 = \bar{\bar{z}}_{\bar{\eta}}^T (\Delta_\eta - \Delta_{\bar{\eta}})^2 \bar{\bar{z}}_{\bar{\eta}}\\
    &=\textstyle \sum_{i=1}^K (\eta_i - \bar{\eta}_i)^2 ([\bar{\bar{z}}_{\bar{\eta}}]_i)^2 \ge \min_i([\bar{\bar{z}}_{\bar{\eta}}]_i)^2 \sum_{i=1}^K (\eta_i - \bar{\eta}_i)^2\\
    &=\min_i([\bar{\bar{z}}_{\bar{\eta}}]_i)^2  \|\eta - \bar{\eta}\|^2,
\end{align*}
where $[\bar{\bar{z}}_{\bar{\eta}}]_i>0$ since $[\bar{\bar{s}}_{\bar{\eta}}]_i=[\Delta_\eta \bar{\bar{z}}_{\bar{\eta}}+\theta_1\mathds{1}]_i>0$ and $\theta_1>0$.

From Lemma \ref{lem:sbm}, we can then use the relation $\| \bar{\bar{s}}_{\eta} - \bar{\bar{s}}_{\bar{\eta}} \|  \le \min_k(\pi_k)^{-\frac{1}{2}} \| \bar{s}_{\eta} - \bar{s}_{\bar{\eta}} \|_{L^2}$ to obtain 
\begin{align*}
    &\|\eta - \bar{\eta}\| \le 2 \min_i(\bar{\bar{z}}_{\bar{\eta}}^i)^{-1} \| \bar{\bar{s}}_{\eta} - \bar{\bar{s}}_{\bar{\eta}} \| \\
    & \textstyle \le \frac{2}{\min_i(\bar{\bar{z}}_{\bar{\eta}}^i) \sqrt{\min_k(\pi_k)}}   \| \bar{s}_{\eta} - \bar{s}_{\bar{\eta}} \|_{L^2} 
    =: L_{\bar{\eta}} \| \bar{s}_{\eta} - \bar{s}_{\bar{\eta}} \|_{L^2}.
\end{align*}

\noindent \textbf{Proof of Proposition \ref{lem:smooth_lq_sbm}}

Since, by Lemma \ref{lem:sbm}, the graphon equilibrium is piecewise constant with values corresponding to the components of $\bar{\bar{s}}_\eta=(I-\Delta_\eta Q \Delta_\pi)^{-1}\theta_1 \mathds{1}$, it suffices to show that the latter is twice Lipschitz continuously differentiable in $\eta$.

To prove Lipschitz continuity, let $\mathbb{V}_\eta := I - \Delta_\eta Q\Delta_\pi$, then $\bar{\bar{s}}_\eta = \theta_1 \mathbb{V}_\eta^{-1} \mathds{1}$. 
We have that $\| \Delta_\eta Q\Delta_\pi\| \le \| \Delta_\eta \| \|Q \Delta_\pi\| \le \max_k(\eta_k) \lambda_{\rm max}(Q \Delta_\pi)= \max_k(\eta_k) \lambda_{\rm max}(\mathbb{W}) < 1$ by Assumption \ref{a:feas_eta} and \cite[Lemma 10]{parise2019graphon} since $Q$ is symmetric. 
Then, from \cite[Theorem 2.3.1]{han2009theoretical}, we get
\begin{align*}
    \|\mathbb{V}_\eta^{-1}\| &= \| (I - \Delta_\eta Q\Delta_\pi)^{-1}\| 
    \le \frac{1}{1-\|\Delta_\eta Q\Delta_\pi\|}\\
    &\le \frac{1}{1 - \max_k(\eta_k) \lambda_{\rm max}(\mathbb{W})}
    \le \frac{1}{1-\alpha^*}
    =: \Bar{V}
\end{align*}
where we used $\max_{\eta \in \Xi} \max_k(\eta_k) \le \eta_{\bar{k}}^*$ for some $\bar{k}\in\{1,...,K\}$ and $\eta^* \in \Xi$, since $\Xi$ is compact and the fact that by Assumption \ref{a:feas_eta}, $\alpha^* :=\eta_{\bar{k}}^* \lambda_{\rm max}(\mathbb{W})<1$. 

By using this fact, it can be easily shown that under the given assumptions, the partial derivatives  of $\bar{\bar{s}}_\eta$ exist\footnote{The partial derivatives of $\bar{\bar{s}}_\eta=\theta_1\mathbb{V}^{-1}_\eta\mathds{1}$ with respect to $\eta$ can be computed by using the identity $\frac{\partial \mathbb{V}_\eta^{-1}}{\partial \eta}=-\mathbb{V}_\eta^{-1}\frac{\partial \mathbb{V}_\eta}{\partial \eta} \mathbb{V}_\eta^{-1}$, \cite[Eq (59)]{petersen2008matrix}.} and can be bounded uniformly in $\eta$
as follows
 \begin{align*}
  \left \| \frac{\partial  \bar{\bar{s}}_\eta}{\partial \eta_i} \right\|  &\le \theta_1  \bar{V}^2  \lambda_{\rm max}(\mathbb{W})  \sqrt{K} =:M_1\\
    \left \| \frac{\partial^2 \bar{\bar{s}}_\eta}{\partial \eta_j \partial \eta_i} \right\| &\le 2  \theta_1  \bar{V}^3  \lambda_{\rm max}(\mathbb{W})^2  \sqrt{K}=:M_2\\
    \left \| \frac{\partial^3 \bar{\bar{s}}_\eta}{\partial \eta_l \eta_j \eta_i} \right\| &\le 6  \theta_1  \bar{V}^4  \lambda_{\rm max}(\mathbb{W})^3  \sqrt{K}=:M_3.
\end{align*}
The conclusion then follows from Lemmas \ref{lem:bded_grad_2} and \ref{lem:sbm}.
We illustrate this result for the equilibrium. 
First, note that the gradient of $\bar{\bar{s}}_\eta$ can be bounded uniformly in $\eta$ as follows
\begingroup
\allowdisplaybreaks
\begin{align*}
    &\|\nabla_\eta \bar{\bar{s}}_\eta\|_2 \le \| \nabla_\eta \bar{\bar{s}}_\eta \|_F 
    =  \sqrt{\sum_{i=1}^K \left \| \frac{\partial \bar{\bar{s}}_\eta}{\partial \eta_i}  \right \|^2_2}\le \sqrt{K}M_1=: L_1
\end{align*}
\endgroup
where $\|\cdot\|_F$ denotes the Frobenius norm. 
For any $k$, this implies that $\| \nabla_\eta [\bar{\bar{s}}_\eta]_k\| \le \| \nabla_\eta \bar{\bar{s}}_\eta \|_2 \le L_1$. 
Hence, by Lemma~\ref{lem:bded_grad_2}, $[\bar{\bar{s}}_\eta]_k$ is Lipschitz continuous with constant $L_1$.
It follows that the graphon equilibrium $\bar{s}_\eta$ is uniformly Lipschitz in $\eta$ since for any $x\in[0,1]$, there exists $ k$ such that
\begin{align*}
    | \bar{s}_\eta(x) - \bar{s}_{\tilde{\eta}}(x) | &= | [\bar{\bar{s}}_\eta]_k - [\bar{\bar{s}}_{\tilde{\eta}}]_k |
    \le L_1 \|\eta - \tilde{\eta} \|. 
\end{align*}
Similar arguments apply to the first and second order derivatives of $\bar{s}_\eta$. 
\endgroup


\bibliographystyle{ieeetr}
\bibliography{references}

\end{document}